\documentclass[11pt]{article}
\usepackage{typearea}\typearea{12}
\usepackage{amsmath,amsfonts,amssymb}
\usepackage{hyperref}
\hypersetup{hidelinks,
 pdftitle={The equivalence of topological indices for U(1) semidirect Z2-symmetric quantum spin chains},
 pdfauthor={Hal Tasaki}}

\newtheorem{theorem}{Theorem}[section]
\newtheorem{lemma}[theorem]{Lemma}
\newtheorem{corollary}[theorem]{Corollary}
\newtheorem{assumption}[theorem]{Assumption}
\newtheorem{proposition}[theorem]{Proposition}

\makeatletter
\@addtoreset{equation}{section}
\makeatother

\def\eq#1\en{\begin{equation}#1\end{equation}}
\def\eqa#1\ena{\begin{align}#1\end{align}}
\def\eqg#1\eng{\begin{gather}#1\end{gather}}
\newcommand{\lb}[1]{\label{e:#1}}
\newcommand{\rlb}[1]{\eqref{e:#1}}
\newcommand{\nl}{\notag\\}

\newcommand{\norm}[1]{\left\Vert#1\right\Vert}
\newcommand{\snorm}[1]{\Vert#1\Vert}

\newcommand{\bbC}{\mathbb{C}}
\newcommand{\bbR}{\mathbb{R}}
\newcommand{\bbZ}{\mathbb{Z}}
\newcommand{\bbN}{\mathbb{N}}
\newcommand{\calA}{\mathfrak{A}}
\newcommand{\calH}{\mathcal{H}}
\newcommand{\calK}{\mathcal{K}}

\newcommand{\calD}{\mathcal{D}}
\newcommand{\Aloc}{\calA_{\rm loc}}
\newcommand{\AL}{\calA_{\rm L}}
\newcommand{\AR}{\calA_{\rm R}}
\newcommand{\mrx}{\mathrm{x}}
\newcommand{\mry}{\mathrm{y}}
\newcommand{\mrz}{\mathrm{z}}
\newcommand{\Sx}{S^{\mrx}}
\newcommand{\Sy}{S^{\mry}}
\newcommand{\Sz}{S^{\mrz}}
\newcommand{\bsS}{\boldsymbol{S}}
\newcommand{\IndT}{\operatorname{Ind}_{\rm T}}
\newcommand{\IndO}{\operatorname{Ind}_{\rm O}}
\newcommand{\Ind}{\operatorname{Ind}}
\newcommand{\Ad}{\operatorname{Ad}}

\newcommand{\id}{\operatorname{id}}
\newcommand{\one}{\hat{1}}
\newcommand{\slim}{\operatorname*{s-lim}}
\newcommand{\Uone}{\mathrm{U}(1)}

\newcommand{\Gsym}{\Uone\rtimes\bbZ_2}
\newcommand{\qedm}{\rule{1.5mm}{3mm}}

\begin{document}

\noindent
{\Large\bf
The equivalence of topological indices for
$\Gsym$-symmetric quantum spin chains}

\renewcommand{\thefootnote}{\fnsymbol{footnote}}
\medskip\noindent
Hal Tasaki\footnote{Department of Physics, Gakushuin University,
Mejiro, Toshima-ku, Tokyo 171-8588, Japan.
\\E-mail: \texttt{hal.tasaki@gmail.com}.}
\renewcommand{\thefootnote}{\arabic{footnote}}
\setcounter{footnote}{0}

\begin{quotation}
\small\noindent
Symmetry-protected topological (SPT) phases in quantum spin chains
are distinguished by topological indices.
Ogata's operator-algebraic indices apply to a broad range of symmetry
classes without requiring continuous symmetry, whereas the elementary
twist index of Tasaki requires $\Uone$ symmetry.
We prove their equivalence in the common setting of integer-spin
chains with on-site $\Gsym$ symmetry and a symmetric locally-unique
gapped ground state, without assuming translation invariance or a
matrix product representation.
A charge-fluctuation estimate yields a strong-limit construction of
half-chain rotations with their phases fixed.
An exact group commutator identity then identifies the Ogata index with the
limit of local twist expectations and gives an explicit error bound.
This provides a computable formula for the index when a lower bound
on the gap and the relevant ground-state expectations are available.
Combining the equivalence with a recent boundary-field gap theorem
of OpenAI and our earlier twist-index calculation, we obtain a
spin-rotation-invariant locally-unique gapped ground state of the
infinite $S=1$ antiferromagnetic Heisenberg chain with both
indices equal to $-1$, without assuming infinite-volume uniqueness.

The author formulated the problem and proposed the basic strategy,
while ChatGPT Astra made substantial contributions to the development of
the proof.
\end{quotation}

\setcounter{tocdepth}{2}
\enlargethispage{2\baselineskip}
\tableofcontents
\medskip

\section{Introduction}\label{S:introduction}
Haldane's prediction of a unique gapped ground state in integer-spin
antiferromagnetic chains \cite{Haldane1981,Haldane1983a,Haldane1983b},
and the subsequent work it inspired, opened the way to the study of
quantum phases not characterized by conventional symmetry breaking.
The Affleck--Kennedy--Lieb--Tasaki (AKLT) models provided rigorously
tractable examples of such gapped ground states \cite{AKLT1988}.
The AKLT construction was also one of the motivations for the work
of Fannes, Nachtergaele, and Werner \cite{FNW1989,FNW1992}, who
introduced matrix product states under the name ``finitely correlated
states'' and proved fundamental theorems about their structure and
properties.
The effective boundary spins, hidden antiferromagnetic order, and
hidden symmetry breaking of the AKLT ground states were early indications of
the structure underlying the Haldane phase
\cite{AKLT1988,denNijsRommelse,KennedyTasaki1992,Oshikawa1992}.

Gu and Wen subsequently identified the Haldane phase as a
symmetry-protected topological (SPT) phase \cite{GuWen2009}.
Pollmann, Turner, Berg, and Oshikawa clarified its protecting
symmetries and characterized it by projective representations on
the auxiliary space of a matrix product state
\cite{Pollmann2010,Pollmann2012}.
Together with the analysis of string order and symmetry
\cite{PG2008} and the classification results based on matrix product
states \cite{ChenGuWen2011,Schuch2011}, these developments established
a systematic framework for understanding one-dimensional SPT phases.

To characterize SPT phases of general infinite-volume ground states,
the indices must be defined without assuming an exact matrix product
representation.
An earlier approach in this direction was developed by Bachmann and
Nachtergaele \cite{BachmannNachtergaele2014}, who studied boundary
symmetry representations as invariants of gapped phases and
constructed excess-spin operators for matrix product states and
for a class of antiferromagnetic chains with a loop representation.
The two bulk-state indices compared in the present paper were
introduced in 2018: the twist index of Tasaki \cite{Tasaki2018} and,
subsequently, the operator-algebraic indices of Ogata
\cite{Ogata2020}.
Their definitions require neither a matrix product representation
nor a loop representation of the ground state.
They differ substantially in their symmetry requirements and in
how their values can be determined.

Ogata's operator-algebraic approach provides the more general
framework. It extends the projective-representation indices of
Pollmann, Turner, Berg, and Oshikawa beyond matrix product states
for a broad range of symmetry classes.
Matsui's work on the split property and symmetry of spin chains
\cite{Matsui2001}, and his theorem deriving the split property from
a bulk gap \cite{Matsui2013}, provide the essential foundation.
Entanglement bounds play a central role in this development
\cite{Matsui2013,Hastings2007}.
For a pure symmetric split state, half-chain symmetries are
implemented by a projective unitary representation.
Ogata used this structure to define topological indices and prove
their invariance under symmetric gapped deformations
\cite{Ogata2020}; a classification of pure symmetric split states
was obtained in \cite{Ogata2021}.
The construction applies to arbitrary finite on-site symmetry groups,
including purely discrete groups.
The operator-algebraic framework also provides indices for
time-reversal and reflection symmetries; see \cite{OgataSPT,OgataICM} for overviews.
For the related theory of automorphic equivalence within gapped
phases, see \cite{BMNS2012,MoonOgata2020}.

By contrast, the twist index $\IndT$ has a more restricted scope,
but admits an elementary expression in terms of the expectation
value of a local observable.
It requires a continuous $\Uone$ rotation symmetry together with a
suitable discrete symmetry.
The twist operator of Lieb, Schultz, and Mattis \cite{LSM1961}, and
its local version used by Affleck and Lieb \cite{AffleckLieb1986},
produce states with small excess energy.
Nakamura and Todo observed that twist expectations distinguish
valence-bond-solid states \cite{NakamuraTodo2002}.
Building on these observations, Tasaki used a local twist to define
a $\bbZ_2$ index for general gapped ground states and prove its
invariance \cite{Tasaki2018}.
The discrete symmetry makes the expectation real, and the gap
prevents it from vanishing for a sufficiently slow twist
\cite{TasakiLocalLSM,Tasaki2023}.

This construction has a computational advantage, not only a
conceptual one: the local observable whose expectation determines
the index is given explicitly.
Huang, Kueng, Torlai, Albert, and Preskill exploited this feature in
their study of efficient learning of quantum phases
\cite[Appendix~K, Theorem~8]{Huang2022}.
For finite periodic $O(2)$-symmetric spin-one chains with controlled
interaction bounds and a common positive gap lower bound, they used
a local twist to obtain a finite-support observable whose expectation
has the same sign as the global twist and stays a fixed distance
from zero.
This connects the twist construction to phase classification from
local measurement data.

In this paper we restrict the operator-algebraic construction to
integer-spin chains with on-site $\Gsym$ symmetry, for which the
twist construction is available, and prove that the two indices
coincide.
We write $\IndO$ for the \emph{Ogata index} or the \emph{operator-algebraic index}.
Here $\IndO$ is the projective commutator sign for the subgroup
$\bbZ_2\times\bbZ_2$ of $\pi$ rotations about two orthogonal axes.
The comparison concerns this common symmetry class; it does not
extend the twist construction to all symmetry classes covered by
the operator-algebraic approach.
The equivalence makes the Ogata index computable from a
finite-interval expectation in the infinite-volume ground state,
given a lower bound on the gap, with an explicit error bound.
It requires neither constructing the half-chain representation nor
finding a gapped interpolation to an exactly solvable model.
The separate tasks of preparing the ground state, certifying the
gap, and controlling finite-volume errors are not addressed by this
statement.

The application to the standard Heisenberg chain has changed since
the first version of this paper.
In two recent preprints, OpenAI proves a uniform spectral gap for
the $S=1$ antiferromagnetic Heisenberg chain on even periodic
rings \cite{OpenAIPeriodic} and, separately, on odd open chains with
equal endpoint fields \cite{OpenAIBoundary}.\footnote{At the time of
this revision, I am studying these proofs in detail. I have so far
found no logical gap in the portions I have examined, but my
verification is not yet complete.}
The second result, with field strength $h_0=3/5$, establishes the
boundary-field gap hypothesis used in our twist-index calculation \cite{Tasaki2025}.
The remaining point is to prove that a boundary-selected local limit
has the spin-rotation symmetry required by the equivalence theorem,
even though the endpoint fields break that symmetry.
A variance estimate, obtained by the same double-commutator argument
as in the main proof, supplies this step.
We thus obtain a spin-rotation-invariant locally-unique gapped
ground state of the infinite $S=1$ Heisenberg chain with both
indices equal to $-1$, without assuming uniqueness among all
infinite-volume ground states.
Only this application uses the new gap theorem of
\cite{OpenAIBoundary}; the general equivalence theorem and its proof
are independent of it. The periodic gap result of
\cite{OpenAIPeriodic} provides important context but is not needed
for this application.

The equivalence proof uses regularized half-chain charges of the
type considered by Bachmann and Nachtergaele
\cite{BachmannNachtergaele2014}.
The bulk gap bounds the fluctuations of the difference between two
such charges. This makes the corresponding rotation unitaries
Cauchy on the GNS ground-state vector and yields strong convergence
with their phases fixed.
An exact identity then relates a half-chain projective commutator
to a one-sided twist.
We state the results in Section~\ref{S:setting}, prove the main
theorem in Section~\ref{S:proof}, and discuss extensions in
Section~\ref{S:discussion}.
The existence and symmetry of the Heisenberg ground state are proved
in Appendix~\ref{S:Heis-appendix}.

\medskip\noindent
{\em Remark:}\/
Within the symmetry class considered here, the equivalence can also
be deduced from the strong-limit construction of Bachmann and
Nachtergaele \cite[Theorems~3.1 and 3.4]{BachmannNachtergaele2014}
for the matrix product states and random-loop models treated there.
Indeed, their regularized half-chain rotations at angle $2\pi$
can be identified, up to taking adjoints and truncating the
exponential cutoff, with local twists whose profiles differ from
the one used here.
Combining their convergence result with the group commutator
identity and the usual twist-deformation argument gives the
equality of the two indices.
The present proof establishes the required strong convergence
directly from the bulk gap and the stated symmetries, without
assuming a matrix product or random-loop representation.

\section{Setting and main results}\label{S:setting}
\subsection{The spin chain and its ground state}
We use the standard operator-algebraic formulation of quantum spin
systems on the infinite chain; see, e.g., Bratteli and Robinson
\cite{BR1,BR2} and Naaijkens \cite{Naaijkens2017}.

Fix an integer $S\ge1$ and associate $\bbC^{2S+1}$ with each site of
$\bbZ$.
For finite $\Lambda\subset\bbZ$, let
$\calA_\Lambda=\bigotimes_{j\in\Lambda}M_{2S+1}(\bbC)$, with the
natural inclusions obtained by tensoring with identities.
The local algebra $\Aloc$ is their union, and $\calA$ is its norm
completion. For arbitrary $\Gamma\subset\bbZ$, $\calA_\Gamma$ denotes
the corresponding quasi-local subalgebra; intervals such as $(0,\infty)$ in subscripts
are understood to be intersected with $\bbZ$.
Write $\bsS_j=(\Sx_j,\Sy_j,\Sz_j)$ for the spin operators, normalized
by $\bsS_j^2=S(S+1)\one$.
In particular, $\snorm{\Sz_j}=S$ and
\eq
e^{-2\pi i\Sz_j}=\one.
\lb{integer}
\en

For local $A\in\calA_\Lambda$, define the global rotation by
\eq
 \beta_{\mrz,\theta}(A)
   =\Ad\biggl(\prod_{j\in\Lambda}e^{-i\theta\Sz_j}\biggr)(A),
 \qquad \Ad(U)(A)=UAU^*.
 \lb{beta-z}
\en
Let $\beta_{\mrx}$ be the analogous $\pi$ rotation about the $\mrx$
axis. The relations $\beta_{\mrz,\theta+2\pi}=\beta_{\mrz,\theta}$,
$\beta_{\mrx}^2=\id$, and
$\beta_{\mrx}\beta_{\mrz,\theta}\beta_{\mrx}=\beta_{\mrz,-\theta}$
define the on-site action of $\Gsym$.
The commuting involutions $\beta_{\mrx}$ and
$\beta_{\mrz}:=\beta_{\mrz,\pi}$ generate
$D_2\cong\bbZ_2\times\bbZ_2$.
Reflection and time reversal are not included in this notation.

We consider a uniformly bounded finite-range interaction, written
as the formal Hamiltonian
\eq
 H=\sum_{j\in\bbZ}h_j,\qquad
 h_j=h_j^*\in\calA_{[j-r,j+r]},\qquad \snorm{h_j}\le J,
 \lb{H}
\en
with fixed $r\in\bbN$ and $J>0$, such that
\eq
\biggl[h_j,\sum_{k=j-r}^{j+r}\Sz_k\biggr]=0,
 \qquad \beta_{\mrx}(h_j)=h_j.
 \lb{local-symmetry}
\en
Translation invariance is not assumed.
For local $A$, the commutator $[H,A]=\sum_j[h_j,A]$ is a finite sum.
A state $\omega$, namely a positive normalized linear functional on
$\calA$, is a ground state if $\omega(A^*[H,A])\ge0$ for every
$A\in\Aloc$.
It is then stationary: $\omega([H,A])=0$.
Our gap condition is
\eq
 \omega(A^*[H,A])\ge
 \gamma\{\omega(A^*A)-|\omega(A)|^2\},\qquad A\in\Aloc,
 \lb{gap}
\en
for some $\gamma>0$.
A ground state satisfying \rlb{gap} is called a \emph{locally-unique
gapped ground state}; see \cite[Section~2.1]{TasakiLSMReview}.
In the Gelfand--Naimark--Segal (GNS) representation, this says that the
ground-state vector spans the zero-energy space of the bulk Hamiltonian
and the remaining spectrum lies in $[\gamma,\infty)$.
A locally-unique gapped ground state is automatically pure; see
\cite[Theorem~A.3]{TasakiLSMReview}.
Local uniqueness does not assert uniqueness among all infinite-volume
ground states.

\begin{assumption}\label{A:state}
The state $\omega$ is a locally-unique gapped ground state of
\rlb{H}, with gap lower bound $\gamma$ as in \rlb{gap}, and is
invariant under $\beta_{\mrz,\theta}$ for all $\theta\in\bbR$ and
under $\beta_{\mrx}$.
\end{assumption}
Symmetry holds automatically for a unique ground state of the
symmetric interaction. We assume it explicitly here because local
uniqueness alone does not imply symmetry.
Spin-flip symmetry implies
\eq
\omega(\Sz_j)=0,\qquad j\in\bbZ.
 \lb{zero-mean}
\en

\subsection{The Ogata index and the twist index}
\label{s:indices}
Let $(\calH,\pi,\Omega)$ be the GNS triple of $\omega$, and put
$\AL=\calA_{(-\infty,-1]}$ and $\AR=\calA_{[0,\infty)}$.
For $g\in D_2$, let $\beta_g^{\rm R}$ act as the identity on $\AL$
and as $\beta_g$ on $\AR$.
We define $\beta_{\mrz,\theta}^{\rm R}$ analogously for every
$\theta\in\bbR$.
Matsui's theorem gives the split property \cite{Matsui2013}.
For a pure state, the split property is equivalent to a tensor-product
realization of the GNS representation: there are irreducible
representations of $\AL$ and $\AR$ on Hilbert spaces
$\calK_{\rm L}$ and $\calK_{\rm R}$ whose tensor product is unitarily
equivalent to the GNS representation; see
\cite[Proposition~1.18]{OgataSPT}.
Since the bicommutant of an irreducible representation is the full
algebra of bounded operators, we may identify
$\calH\cong\calK_{\rm L}\otimes\calK_{\rm R}$ so that
\eq
 \pi(\AL)''=B(\calK_{\rm L})\otimes\one,\qquad
 \pi(\AR)''=\one\otimes B(\calK_{\rm R}).
 \lb{split}
\en
Here $B(\calK)$ denotes the algebra of bounded operators on $\calK$.
The double prime denotes the bicommutant; the resulting algebras are
von Neumann algebras.
The GNS vector $\Omega$ need not be a product vector in this decomposition.
The symmetry of $\omega$ induces a normal automorphism of each
half-chain algebra; every automorphism of $B(\calK_{\rm R})$ is inner.
There are therefore unitaries $u_g$ such that
\eq
 u_g\pi(A)u_g^*=\pi(\beta_g^{\rm R}(A)),\qquad A\in\calA.
 \lb{implementer}
\en
They can be chosen as $\one\otimes v_g$ in \rlb{split}, so this
whole-chain GNS formulation gives the same index as the half-chain
factor formulation \cite{Ogata2020,Ogata2021,OgataICM}.

Purity implies $\pi(\calA)'=\bbC\one$, so the implementers are unique
up to phase and satisfy $u_gu_h=\sigma(g,h)u_{gh}$ with
$\sigma(g,h)\in\Uone$.
The two classes of $H^2(D_2,\Uone)\cong\bbZ_2$ are distinguished by
\eq
 u_{\mrz}u_{\mrx}u_{\mrz}^*u_{\mrx}^*
       =\IndO(\omega)\one,\qquad \IndO(\omega)\in\{-1,1\}.
 \lb{IndO}
\en
The group commutator implements the identity and is a scalar; since
$u_{\mrz}^2$ is also a scalar, the commutator scalar has square one.
This phase-independent sign is the \emph{Ogata index}.
Our proof uses only the abstract existence of $u_{\mrx}$; it constructs
$u_{\mrz}$ from local rotations.

For $N\in\bbN$, define the twist operator
\eq
 T_N=\exp\biggl[-\frac{2\pi i}{N}\sum_{j=0}^{N}j\Sz_j\biggr].
 \lb{one-sided-twist}
\en
The twist operator rotates spins on sites $j=1,\ldots,N-1$.
Its endpoint factors at $j=0$, $N$ are the identity by \rlb{integer}.
Since $\beta_{\mrx}(T_N)=T_N^*$, spin-flip symmetry implies that
$\omega(T_N)$ is real.
The twist-energy estimate and gap make its magnitude tend to one and
its sign constant for large $N$ \cite{Tasaki2018}.
The \emph{twist index} is then defined as
\eq
 \IndT(\omega)=\lim_{N\uparrow\infty}\omega(T_N)\in\{-1,1\}.
 \lb{IndT}
\en
We will also verify this limit in the present setting.

\medskip\noindent
{\em Remark:}\/ In our previous works \cite{Tasaki2018,Tasaki2025},
we used twist profiles centered around the origin, but they give the
same index as \rlb{IndT}.
Indeed, because $e^{-2\pi i\Sz_j}=\one$ at every site, shifting the
twist profile continuously along the chain gives a norm-continuous
deformation, while the usual twist-energy bound remains uniformly
small when the twist is spread over sufficiently many sites.
The gap and spin-flip symmetry then prevent its real expectation from
changing sign. This is the deformation argument of \cite{Tasaki2018};
independence of the twist profile is stated explicitly in
\cite[Theorem~3.1 and Section~4.1]{Tasaki2023} for the fermionic index.
No translation invariance of the state is required.

\subsection{Equivalence and examples}
The following is our main theorem.
\begin{theorem}\label{T:main}
For the interaction \rlb{H} satisfying \rlb{local-symmetry} and a
state $\omega$ satisfying Assumption~\ref{A:state},
\eq
 \IndT(\omega)=\IndO(\omega),\qquad
 \slim_{N\uparrow\infty}\pi(T_N)=\IndO(\omega)\one.
 \lb{equivalence}
\en
\end{theorem}

Henceforth we write $\Ind(\omega)$ for the common value
$\IndT(\omega)=\IndO(\omega)$.
We now review some properties of $\Ind(\omega)$.
The following quantitative estimate is a version of the local-twist
bounds in \cite{Tasaki2018,Tasaki2023}; our proof yields it in the
explicit form
\eq
 0\le1-\Ind(\omega)\,\omega(T_N)
                  \le\frac{2\pi^2\kappa}{\gamma N},
 \lb{twist-error}
\en
for every $N\ge1$, with the nonoptimal constant
\eq
\kappa=4JS^2r^2(2r+1)^2.
\lb{kappa}
\en
To express the expectation as that of a self-adjoint observable, put
$A_N=(T_N+T_N^*)/2$.
This is the same passage to the Hermitian part used in
\cite[Appendix~K, Theorem~8]{Huang2022}; it introduces no new index.
Since $\omega(T_N)\in\bbR$, we have $\omega(A_N)=\omega(T_N)$, so
\eq
 \Ind(\omega)\,\omega(A_N)\ge\frac12
 \quad\text{if}\quad N\ge\frac{4\pi^2\kappa}{\gamma}.
 \lb{classification-margin}
\en
An estimate of the expectation with error less than $1/2$ therefore
determines $\Ind(\omega)\in\{-1,1\}$.
The required twist length depends only on $S,r,J$, and a known gap
lower bound. This is the local-observable sense of computability:
the expectation is taken in the infinite-volume state $\omega$.
Replacing it by a finite-volume expectation requires separate
control of that approximation.

For example, the trivial Hamiltonian $H_{\rm tr}=\sum_j(\Sz_j)^2$ has a unique product ground
state $\omega_{\rm tr}$ with gap one and $\Ind(\omega_{\rm tr})=1$.
The $S=1$ AKLT model
\eq
 H_{\rm AKLT}=\sum_j\Bigl\{
 \bsS_j\cdot\bsS_{j+1}
 +\frac13(\bsS_j\cdot\bsS_{j+1})^2\Bigr\}
\en
has a unique gapped infinite-volume ground state $\omega_{\rm AKLT}$
\cite{AKLT1988,Matsui1997} with $\Ind(\omega_{\rm AKLT})=-1$
\cite{Tasaki2018}.
The same sign describes the anticommuting spin-$\frac12$ boundary
rotations \cite{Pollmann2010,BachmannNachtergaele2014,Ogata2020}.
Anisotropic and spatially nonuniform interactions are also covered
whenever the symmetry and gap hypotheses hold.

\subsection{Application to the \texorpdfstring{$S=1$}{S=1} Heisenberg chain}\label{S:Heis}
We now specialize to $S=1$ and consider the standard
antiferromagnetic Heisenberg chain and its finite open-chain version
with endpoint fields:
\eqg
 H_{\rm Heis}=\sum_{j\in\bbZ}\bsS_j\cdot\bsS_{j+1},\qquad
 H_{L,h}=\sum_{j=-L}^{L-1}\bsS_j\cdot\bsS_{j+1}
                    -h(\Sz_{-L}+\Sz_L),
 \lb{Heis}
\eng
where the finite chain consists of $2L+1$ sites.
Let $E_0(L,h)$ and $E_1(L,h)$ be the lowest and second-lowest
energies of $H_{L,h}$, counted with multiplicity.
For $h>0$, the finite-volume ground state is unique
\cite[Lemma~1]{Tasaki2025}.
We use the following slightly weakened form of
\cite[Theorem~1.1]{OpenAIBoundary}, with rounded constants.

\begin{theorem}[OpenAI \cite{OpenAIBoundary}]\label{T:boundary-gap}
For $S=1$, let $h_0=0.6$, $\gamma_0=5\times10^{-3}$, and $L_0=10^3$.
Then
\eq
 E_1(L,h_0)-E_0(L,h_0)\ge\gamma_0,\qquad L\ge L_0.
 \lb{Heis-gap}
\en
\end{theorem}

This verifies the essential boundary-field gap hypothesis in our earlier work
\cite{Tasaki2025} for $S=1$.
Let $\omega_L$ be the ground state of $H_{L,h_0}$.
By compactness, there is an increasing sequence $L_n\uparrow\infty$
along which the ground states have a local limit
\eq
 \omega_{\rm Heis}(A)=\lim_{n\uparrow\infty}\omega_{L_n}(A),
 \qquad A\in\Aloc.
 \lb{Heis-limit}
\en
The following proposition applies to any limit obtained in this way.

\begin{proposition}\label{P:Heis-state}
The state $\omega_{\rm Heis}$ in \rlb{Heis-limit} is a
locally-unique gapped ground state of $H_{\rm Heis}$, with gap lower
bound $\gamma_0$, and is invariant under all global spin rotations.
\end{proposition}

The passage of the gap to the local limit is standard; see
\cite[Theorem~2.6]{TasakiLSMReview}.
The symmetry assertion is not an immediate consequence of the
finite-volume symmetry: the endpoint fields break rotations about
the $\mrx$ and $\mry$ axes, and
$\beta_{\mrx}(H_{L,h_0})=H_{L,-h_0}$.
Appendix~\ref{S:Heis-appendix} gives the proof, deriving rotational
invariance from a bound on charge variances without assuming that
the state is symmetric.

\begin{corollary}\label{C:Heis}
The state $\omega_{\rm Heis}$ in Proposition~\ref{P:Heis-state}
satisfies
\eq
 \Ind(\omega_{\rm Heis})=-1.
 \lb{Heis-index}
\en
In particular, the infinite $S=1$ antiferromagnetic Heisenberg
chain has a spin-rotation-invariant locally-unique gapped ground
state whose twist and Ogata indices are both $-1$.
\end{corollary}
\noindent{\em Proof:}\/ Theorem~\ref{T:boundary-gap} supplies the
finite-volume gap hypothesis of \cite{Tasaki2025}, whose twist-index
calculation gives $-1$ for the limiting state
\rlb{Heis-limit}, using a twist centered at the origin.
By Proposition~\ref{P:Heis-state}, the state has the required
spin-flip symmetry. The Remark at the end of
Section~\ref{s:indices} therefore identifies that value with the
index defined using $T_N$.
The same proposition supplies all the state hypotheses of
Theorem~\ref{T:main}, which proves the assertion.\hfill\qedm

\medskip
The infinite-volume ground state is strongly expected to be unique,
but a proof of uniqueness among all infinite-volume ground states
is not presently known to the author.
Neither Proposition~\ref{P:Heis-state} nor Corollary~\ref{C:Heis}
requires this uniqueness.
Even for the AKLT chain, establishing uniqueness among all
infinite-volume ground states, rather than only within the class of
zero-energy states (as treated in \cite{AKLT1988}), required a separate argument, due to Matsui
\cite{Matsui1997}.

For the Heisenberg chain with integer $S>1$, the same argument gives
$\Ind(\omega_{\rm Heis})=(-1)^S$ for a boundary-selected local limit,
provided the uniform boundary-field gap hypothesis of
\cite{Tasaki2025} holds. 
The variance argument in
Appendix~\ref{S:Heis-appendix} extends to every fixed $S$, so
infinite-volume uniqueness is again unnecessary. The quoted gap
theorem, however, concerns only $S=1$.

\section{Proof of Theorem~\ref{T:main}}\label{S:proof}
For this proof, we distinguish the Ogata index $\IndO(\omega)$
defined by \rlb{IndO} from the twist index $\IndT(\omega)$ defined
by the limit \rlb{IndT}. Our goal is to prove that they coincide.

\subsection{A charge-fluctuation estimate}\label{S:fluctuation}
For a finitely supported real sequence $f=(f_j)_{j\in\bbZ}$, let
\eq
             Q(f)=\sum_jf_j\Sz_j,\qquad
             \calD(f)=\sum_j(f_{j+1}-f_j)^2.
 \lb{Q-D}
\en
The following estimate is the main analytic input.
It expresses the fact that, in a gapped state with a conserved spin
component, slowly varying weighted spin sums have small fluctuations.

\begin{lemma}\label{L:fluctuation}
Under the assumptions of Theorem~\ref{T:main},
\eq
                         \omega(Q(f)^2)
                              \le\frac{\kappa}{\gamma}\,\calD(f)
 \lb{fluctuation}
\en
for every finitely supported real $f$.
\end{lemma}

\noindent{\em Proof:}\/ By \rlb{zero-mean}, $\omega(Q(f))=0$.
The gap inequality then gives
\eqa
 \gamma\,\omega\bigl(Q(f)^2\bigr)
 \le\omega\bigl(Q(f)[H,Q(f)]\bigr)=\frac12\omega\bigl([Q(f),[H,Q(f)]]\bigr),
 \lb{double-commutator}
\ena
where we used the stationarity $\omega([H,Q(f)^2])=0$ to get the final equality.
All expressions in \rlb{double-commutator} are finite sums.

For each local interval $[j-r,j+r]$, define
\eq
                  B_j(f)=\sum_{k=j-r}^{j+r}(f_k-f_j)\Sz_k.
 \lb{Bj}
\en
To justify the replacement in both commutators in \rlb{replace}, put
$C_j=Q(f)-B_j(f)$.
On $[j-r,j+r]$, this operator is $f_j\sum_{k=j-r}^{j+r}\Sz_k$;
its remaining terms are supported outside the interval.
Thus $[C_j,h_j]=0$ by \rlb{local-symmetry} and locality.
Also, $[C_j,B_j(f)]=0$ since all $\Sz_k$ commute.
It follows that $[h_j,Q(f)]=[h_j,B_j(f)]$ and
$[C_j,[h_j,B_j(f)]]=0$, so
\eq
             [Q(f),[h_j,Q(f)]]=[B_j(f),[h_j,B_j(f)]].
 \lb{replace}
\en
To estimate $B_j(f)$, note that $f_k-f_j$ is a sum of $|k-j|$
nearest-neighbor differences, up to an overall sign. Thus, for
$|k-j|\le r$, Cauchy--Schwarz gives
\eq
 |f_k-f_j|^2\le |k-j|
 \sum_{m=\min\{j,k\}}^{\max\{j,k\}-1}(f_{m+1}-f_m)^2
 \le r\sum_{m=j-r}^{j+r-1}(f_{m+1}-f_m)^2.
\en
Note that the final bound is independent of $k$.
The triangle inequality then implies
\eq
 \snorm{B_j(f)}\le S\sum_{k=j-r}^{j+r}|f_k-f_j|
 \le S(2r+1)\biggl\{r\sum_{m=j-r}^{j+r-1}
                         (f_{m+1}-f_m)^2\biggr\}^{1/2},
\en
and hence
\eq
 \snorm{B_j(f)}^2\le
 S^2(2r+1)^2\,r\sum_{m=j-r}^{j+r-1}(f_{m+1}-f_m)^2.
 \lb{Bj-bound}
\en
Since $\snorm{[B,[h,B]]}\le4\snorm{h}\snorm{B}^2$ and each difference
occurs for $2r$ values of $j$, summing \rlb{Bj-bound} yields
\eq
 \frac12\sum_j\snorm{[B_j(f),[h_j,B_j(f)]]}
       \le2J\sum_j\snorm{B_j(f)}^2\le\kappa\,\calD(f),
 \lb{summed-bound}
\en
with $\kappa$ defined in \rlb{kappa}.
Together with \rlb{double-commutator}, this proves the lemma.\hfill\qedm

\medskip
Without \rlb{zero-mean}, the same proof bounds the variance
$\omega(Q(f)^2)-\omega(Q(f))^2$.
Here the spin-flip symmetry makes the variance equal to the uncentered
second moment. This point will fix the phases in the limiting rotation
operators, not only the states obtained by applying them.

\subsection{The strong-limit construction}\label{S:half}
For $N\in\bbN$, define
\eq
 f_N(j)=
 \begin{cases}
  1-j/N,&0\le j\le N,\\
  0,&j<0\text{ or }j>N,
 \end{cases}
 \qquad Q_N=Q(f_N),\qquad W_N(\theta)=e^{-i\theta Q_N}.
 \lb{cutoff}
\en
The unitary $W_N(\theta)$ is a local operator supported in the right
half-chain. Its rotation angle varies linearly from $\theta$ at the
cut to zero at site $N$. At each fixed site in the right half-chain,
the angle tends to $\theta$ as $N\uparrow\infty$.
For each fixed local observable $A$,
\eq
             \Ad(W_N(\theta))(A)\longrightarrow
                         \beta_{\mrz,\theta}^{\rm R}(A)
                         \qquad(N\uparrow\infty)
 \lb{local-limit}
\en
in norm, uniformly for $\theta$ in compact intervals.
Convergence of these automorphisms alone would not control the phases
of the implementing unitaries. We will prove convergence of the
unitaries themselves.

Each $f_N$ has a jump of size one at the bond $(-1,0)$.
This jump cancels exactly in $f_M-f_N$, and one has
\eq
                  \calD(f_M-f_N)=\left|\frac1M-\frac1N\right|.
 \lb{difference-D}
\en
To see this, assume $M\ge N$ and note that the increments of
$f_M-f_N$ are $1/N-1/M$ on the $N$ bonds from $0$ to $N$, and
$-1/M$ on the $M-N$ bonds from $N$ to $M$; all other increments
vanish.
Their squared sum is
$N(1/N-1/M)^2+(M-N)/M^2=1/N-1/M$.
This cancellation at the fixed cut is essential.

\begin{proposition}\label{P:half-rotation}
The strong limits
\eq
                U_{\rm R}(\theta)=\slim_{N\uparrow\infty}\pi(W_N(\theta)),
                \qquad \theta\in\bbR,
 \lb{strong-limit}
\en
exist and form a strongly continuous unitary representation of
$\bbR$ on $\calH$.
They implement the right-half-chain rotations:
\eq
 U_{\rm R}(\theta)\pi(A)U_{\rm R}(\theta)^*
          =\pi(\beta_{\mrz,\theta}^{\rm R}(A)),\qquad A\in\calA.
 \lb{half-action}
\en
Moreover,
\eq
 \norm{\{\pi(W_N(\theta))-U_{\rm R}(\theta)\}\Omega}
                  \le |\theta|\sqrt{\frac{\kappa}{\gamma N}}.
 \lb{rotation-rate}
\en
\end{proposition}

\noindent{\em Proof:}\/ The operators $Q_M$ and $Q_N$ commute.
Using $2-2\cos t\le t^2$, Lemma~\ref{L:fluctuation}, and
\rlb{difference-D}, we find
\eqa
 \norm{\pi(W_M(\theta)-W_N(\theta))\Omega}^2
 &=\omega\bigl(2\one-2\cos(\theta(Q_M-Q_N))\bigr)\nl
 &\le\theta^2\omega((Q_M-Q_N)^2)\nl
 &\le\frac{\kappa\theta^2}{\gamma}
                       \left|\frac1M-\frac1N\right|.
 \lb{Cauchy}
\ena
Thus $\pi(W_N(\theta))\Omega$ is a Cauchy sequence, uniformly for
$\theta$ in compact intervals.

For $A\in\Aloc$, use the identity
\eq
 \pi(W_N(\theta))\pi(A)\Omega
   =\pi(\Ad(W_N(\theta))(A))\pi(W_N(\theta))\Omega.
 \lb{dense-set}
\en
By \rlb{local-limit} and \rlb{Cauchy}, the right-hand side converges,
again locally uniformly in $\theta$.
Since $\pi(\Aloc)\Omega$ is dense in $\calH$ and the unitaries are
uniformly bounded, this proves strong convergence on $\calH$.
Applying the same argument to $-\theta$ proves strong convergence of
the adjoints.
The limits are therefore unitary, and passing to the limit in
$W_N(\theta)W_N(\varphi)=W_N(\theta+\varphi)$ gives the group law.
The locally uniform convergence proves strong continuity.
Passing to the limit in the conjugation identity proves
\rlb{half-action}, first for local $A$ and then for all $A\in\calA$.
Finally, let $M\uparrow\infty$ in \rlb{Cauchy} to obtain
\rlb{rotation-rate}.\hfill\qedm

\medskip
It is important that this is a representation of $\bbR$.
Although the implemented automorphisms are $2\pi$-periodic, the
unitaries themselves need not be. Their value at $2\pi$ will contain
the topological information.
No convergence in operator norm is asserted.

\subsection{Identification with a one-sided twist}
Let $u_{\mrx}$ implement $\beta_{\mrx}^{\rm R}$ as in
\rlb{implementer}.
Since $W_N(\theta)\in\AR$, we have the exact identity
\eq
 u_{\mrx}\pi(W_N(\theta))u_{\mrx}^*=\pi(W_N(-\theta)).
 \lb{flip-local}
\en
Passing to strong limits gives
\eq
 u_{\mrx}U_{\rm R}(\theta)u_{\mrx}^*=U_{\rm R}(-\theta).
 \lb{flip-limit}
\en
Using \rlb{integer}, we obtain
\eq
 W_N(2\pi)=\prod_{j=0}^{N}e^{-2\pi i(1-j/N)\Sz_j}
          =\prod_{j=0}^{N}e^{2\pi i(j/N)\Sz_j}=T_N^*.
 \lb{exact-twist}
\en
Combining \rlb{flip-local} and \rlb{exact-twist}, we obtain
\eq
 \pi(W_N(\pi))u_{\mrx}\pi(W_N(\pi))^*u_{\mrx}^*=\pi(T_N^*).
 \lb{finite-commutator}
\en
This is the identity that connects the two constructions.
By Proposition~\ref{P:half-rotation}, $u_{\mrz}:=U_{\rm R}(\pi)$
implements the required half-chain $\pi$ rotation.
The left-hand side of \rlb{finite-commutator} converges strongly to
the group commutator \rlb{IndO} defining $\IndO(\omega)$.
The right-hand side converges strongly to $U_{\rm R}(2\pi)$ by
\rlb{exact-twist} and \rlb{strong-limit}, and hence
\eq
 U_{\rm R}(2\pi)=\IndO(\omega)\one.
 \lb{two-pi}
\en
The adjoints also converge strongly by Proposition~\ref{P:half-rotation},
which proves the strong limit in \rlb{equivalence}, since
$\IndO(\omega)\in\{-1,1\}$.
Taking its expectation in $\Omega$ establishes the limit in
\rlb{IndT} and proves $\IndT(\omega)=\IndO(\omega)$.
Finally, \rlb{rotation-rate} at $\theta=2\pi$ and the reality of
$\omega(T_N)$ give
\eq
\frac{4\pi^2\kappa}{\gamma N}\ge
\norm{\{\pi(T_N^*)-\IndO(\omega)\one\}\Omega}^2
=
 2\{1-\IndO(\omega)\,\omega(T_N)\},
\en
which proves \rlb{twist-error} and completes the proof of
Theorem~\ref{T:main}.\hfill\qedm

\section{Discussion}\label{S:discussion}
We have shown that the twist and Ogata indices detect the same projective
obstruction under the assumptions of Theorem~\ref{T:main}.
The identity \rlb{finite-commutator} and the strong-limit estimate
\rlb{rotation-rate} give a direct comparison, and the quantitative
bound \rlb{twist-error} makes the abstract index accessible through
local expectations.
For the $S=1$ Heisenberg chain, the boundary-field gap theorem of
OpenAI, the symmetry argument in Appendix~\ref{S:Heis-appendix},
and our earlier twist-index calculation now yield a locally-unique
gapped ground state with $\Ind=-1$.
This conclusion does not require uniqueness among all infinite-volume
ground states; that remains a separate question.
It would be interesting to establish the existence of a
$\Gsym$-invariant locally-unique gapped ground state with
$\Ind=-1$ for the $S=1$ XXZ chain
$H_\lambda=\sum_j\{\Sx_j\Sx_{j+1}+\Sy_j\Sy_{j+1}
+\lambda\Sz_j\Sz_{j+1}\}$,
with $\lambda\ne1$ sufficiently close to $1$.

The comparison can be extended to half-odd-integer spins and, more
generally, to uniformly bounded site-dependent spin quantum numbers
$S_j$. In this setting the common index depends on the cut:
if $\Ind_a(\omega)$ denotes the index at the cut $(a-1,a)$, then
$\Ind_{a+1}(\omega)=(-1)^{2S_a}\Ind_a(\omega)$.
For the twist construction, the integer-spectrum generators
$\Sz_j-S_j\one$ ensure continuity under changes of the twist profile,
since $e^{-2\pi i(\Sz_j-S_j\one)}=\one$.
The scalar phases must be normalized consistently with the chosen
cut. With this convention the twist index again agrees with the Ogata
index. The role of the choice of a unit cell in the related
fermionic twist index is discussed in \cite{Tasaki2023}.
Details of the spin-chain extension will be given in the technical
note \cite{TechnicalNote}.

A further direction is the corresponding comparison for
$[\Uone]^{n-1}\rtimes\bbZ_n$ symmetry, where the cyclic group acts by
permuting the diagonal charges. This is the diagonal-torus and cyclic
permutation subgroup of $\mathrm{PSU}(n)$, and reduces to $\Gsym$
when $n=2$.
Twist expectations with $n$th-root quantization in $\mathrm{SU}(n)$
spin models have already been studied by Su, Yao, and Furusaki
\cite{SuYaoFurusaki}.
The extension of the present comparison to this setting will be
discussed in the same technical note \cite{TechnicalNote}.

\appendix
\section{Proof of Proposition~\ref{P:Heis-state}}\label{S:Heis-appendix}
We first pass the finite-volume gap to a local limit, then prove
spin-rotation invariance. The second step is the one needed because
the endpoint fields break the symmetry.

\subsection{The locally-unique gapped limiting state}
Fix a local operator $A$.
For all sufficiently large $L$, the endpoint terms commute with $A$
and the omitted interaction terms are also disjoint from its support.
Thus $[H_{L,h_0},A]=[H_{\rm Heis},A]$.
The finite-volume gap inequality gives
\eq
 \omega_L\bigl(A^*[H_{\rm Heis},A]\bigr)
 \ge\gamma_0\bigl\{\omega_L(A^*A)-|\omega_L(A)|^2\bigr\}.
 \lb{Heis-finite-gap-ineq}
\en
Passing to the subsequence in \rlb{Heis-limit} yields
\eq
 \omega_{\rm Heis}\bigl(A^*[H_{\rm Heis},A]\bigr)
 \ge\gamma_0\bigl\{\omega_{\rm Heis}(A^*A)
                         -|\omega_{\rm Heis}(A)|^2\bigr\}.
 \lb{Heis-bulk-gap-ineq}
\en
Since $A$ was arbitrary, this is the ground-state condition and
\rlb{gap} with lower bound $\gamma_0$.
In particular, $\omega_{\rm Heis}$ is a locally-unique gapped ground
state and hence is pure.
This argument is standard; see, e.g., \cite[Theorem~2.6]{TasakiLSMReview}.

\subsection{Spin-rotation symmetry from the variance estimate}
For brevity write $\omega=\omega_{\rm Heis}$.
We do not assume any spin-rotation symmetry of $\omega$ in the
following argument.
For $\alpha=\mrx,\mry,\mrz$ and a finitely supported real sequence $f$,
put
\eq
 Q_\alpha(f)=\sum_j f_j S_j^\alpha,\qquad
 \widetilde Q_\alpha(f)=Q_\alpha(f)-\omega(Q_\alpha(f))\one,
 \qquad
 \operatorname{Var}_\omega(Q_\alpha(f))
       =\omega\bigl(\widetilde Q_\alpha(f)^2\bigr).
 \lb{Heis-centered-charge}
\en
Apply \rlb{Heis-bulk-gap-ineq} to $\widetilde Q_\alpha(f)$.
Stationarity removes the scalar subtraction from the commutator
expectation and gives
\eqa
 \gamma_0\operatorname{Var}_\omega(Q_\alpha(f))
 &\le\omega\bigl(Q_\alpha(f)[H_{\rm Heis},Q_\alpha(f)]\bigr)\nl
 &=\frac12\omega\bigl([Q_\alpha(f),
                         [H_{\rm Heis},Q_\alpha(f)]]\bigr).
 \lb{Heis-variance-double}
\ena
Here the last equality follows from
$\omega([H_{\rm Heis},Q_\alpha(f)^2])=0$.

Each Heisenberg interaction term conserves the spin component
$\alpha$ on its support.
Consequently, the operator-norm calculation
\rlb{Bj}--\rlb{summed-bound}, with $\Sz_j$ replaced by $S_j^\alpha$,
applies without any symmetry assumption on the state.
For the $S=1$ Heisenberg interaction we may take $r=1$
and $J=2$ in \rlb{kappa}, so that $\kappa=72$.
It follows from \rlb{Heis-variance-double} that
\eq
 \operatorname{Var}_\omega(Q_\alpha(f))
 \le\frac{72}{\gamma_0}\calD(f),\qquad
 \calD(f)=\sum_j(f_{j+1}-f_j)^2.
 \lb{Heis-variance}
\en
Only the variance is bounded here; no vanishing one-point
expectations have been assumed.

Let $A\in\calA_{[-m,m]}$ be an arbitrary local operator.
For a positive integer $R$, choose
\eq
 f_R(j)=
 \begin{cases}
  1,&|j|\le m,\\
  1-(|j|-m)/R,&m<|j|<m+R,\\
  0,&|j|\ge m+R.
 \end{cases}
 \lb{Heis-plateau}
\en
There are $R$ nonzero increments of magnitude $1/R$ on each side,
so $\calD(f_R)=2/R$.
Since $f_R=1$ on the support of $A$, we also have
\eq
 [Q_\alpha(f_R),A]=[K_\alpha,A],\qquad
 K_\alpha=\sum_{j=-m}^{m}S_j^\alpha.
 \lb{Heis-local-generator}
\en
The right-hand side is independent of $R$.
Subtracting the mean of $Q_\alpha(f_R)$ and applying Cauchy--Schwarz
to the two terms of the commutator gives
\eqa
 \bigl|\omega([K_\alpha,A])\bigr|
 =\Bigl|\omega\bigl([\widetilde Q_\alpha(f_R),A]\bigr)\Bigr|
 \le 2\snorm{A}
          \sqrt{\operatorname{Var}_\omega(Q_\alpha(f_R))}
 \le 2\snorm{A}\sqrt{\frac{144}{\gamma_0 R}}.
 \lb{Heis-vanishing-generator}
\ena
Letting $R\uparrow\infty$ therefore yields
$\omega([K_\alpha,A])=0$.

For a local operator supported in $[-m,m]$, write
$\beta_{\alpha,\theta}(A)=e^{-i\theta K_\alpha}Ae^{i\theta K_\alpha}$.
The rotated operator has the same support, so the preceding identity
applies to it as well. Hence
\eq
 \frac{d}{d\theta}\omega(\beta_{\alpha,\theta}(A))
 =-i\,\omega\bigl([K_\alpha,\beta_{\alpha,\theta}(A)]\bigr)=0.
 \lb{Heis-rotation-invariance}
\en
Thus $\omega\circ\beta_{\alpha,\theta}=\omega$ on the local algebra
and, by norm continuity, on $\calA$.
Rotations about the coordinate axes generate all spin rotations.
This proves Proposition~\ref{P:Heis-state}.\hfill\qedm

\medskip
The argument uses continuous spin-rotation symmetry, which is
available for the isotropic Heisenberg Hamiltonian.
It establishes spin-rotation invariance of the limiting state,
but does not resolve the separate question of uniqueness of the
infinite-volume ground state.

\section*{History of the argument and use of AI}
\addcontentsline{toc}{section}{History of the argument and use of AI}
I had been trying for some time to prove the equivalence of the two
indices, initially using twists centered at the origin, as in
\cite{Tasaki2025}.
I discussed their relation to the half-chain projective
representation with Yoshiko Ogata, but these discussions did not
lead to a proof.
She subsequently proved and communicated to me a strong-convergence
result for the corresponding regularized half-chain unitaries.
This provided an important ingredient, but I still did not see how
to identify the two indices.

About a year ago, I realized that the right approach was to use
twists whose angle starts to increase at the origin, as in
\rlb{one-sided-twist}.
In light of our preceding discussions and Ogata's convergence
result, I expected this one-sided construction to yield the desired
comparison.
Before returning to Ogata with the revised approach, I asked ChatGPT
(GPT-6 Astra Pro, OpenAI) to develop it.
I provided my papers and slides presenting the idea, together with
Ogata's account of her index.
ChatGPT developed the equivalence proof presented in
Section~\ref{S:proof}, including the cutoff-difference estimate,
the strong-limit argument with controlled phases, and the exact
commutator identity.
Its contribution was substantive and mathematical, not limited to
assistance with wording.

I examined the resulting proof before asking ChatGPT to prepare a
manuscript draft.
It also assisted with subsequent substantial revisions and the literature search.

After the first version of this paper was posted on arXiv, OpenAI
made available the two preprints addressing the Haldane gap problem
\cite{OpenAIPeriodic,OpenAIBoundary}.
I asked ChatGPT (GPT-6 Astra Pro, OpenAI) whether the boundary-field
result, which establishes the gap hypothesis of my earlier work,
implies the existence of a locally-unique gapped infinite-volume
ground state with both indices equal to $-1$.
The point requiring an additional argument was the symmetry of the
limiting state, since the endpoint fields break spin-rotation
symmetry.
ChatGPT supplied the proof in Appendix~\ref{S:Heis-appendix},
which derives this symmetry directly from a charge-variance bound
and does not require global uniqueness of the ground state.
It also prepared the corresponding revision of the manuscript.
I have also checked all mathematical arguments added in this revision,
including the proof in Appendix~\ref{S:Heis-appendix}.

I take full responsibility for the mathematical arguments,
references, and final text.

\section*{Acknowledgments}
\addcontentsline{toc}{section}{Acknowledgments}
I wish to thank Yoshiko Ogata, whom I regard as an intellectual
coauthor of this work, for extensive discussions and for proving and
communicating the strong-convergence result described above.
I also thank Sven Bachmann and Bruno Nachtergaele for useful comments.
This work is supported by JSPS Grants-in-Aid for Scientific Research
No.\ 25K07171.



\begin{thebibliography}{99}
\addcontentsline{toc}{section}{References}
\setlength{\itemsep}{0pt}

\bibitem{Haldane1981}
F.D.M. Haldane,
{\em Ground State Properties of Antiferromagnetic Chains with Unrestricted
Spin: Integer Spin Chains as Realisations of the $O(3)$ Non-Linear Sigma
Model}\/,
ILL preprint SP-81/95 (1981).
\\\url{https://arxiv.org/abs/1612.00076}

\bibitem{Haldane1983a}
F.D.M. Haldane, 
{\em Continuum dynamics of the 1-D Heisenberg antiferromagnet: identification with the $O(3)$ nonlinear sigma model}\/,
Phys. Lett. {\bf 93A}, 464--468 (1983).
\\\url{http://www.sciencedirect.com/science/article/pii/037596018390631X}

\bibitem{Haldane1983b}
F.D.M. Haldane, 
{\em Nonlinear field theory of large-spin Heisenberg antiferromagnets: semiclassically quantized solitons of the one-dimensional easy-axis N\'eel state}\/,
Phys. Rev. Lett. {\bf 50}, 1153--1156 (1983).
\\\url{https://journals.aps.org/prl/abstract/10.1103/PhysRevLett.50.1153}

\bibitem{AKLT1988}
I. Affleck, T. Kennedy, E.H. Lieb, and H. Tasaki,
{\em Valence bond ground states in isotropic quantum antiferromagnets}\/,
Commun. Math. Phys. {\bf 115}, 477--528 (1988).
\\\url{https://projecteuclid.org/euclid.cmp/1104161001}

\bibitem{FNW1989}
M. Fannes, B. Nachtergaele, and R. F. Werner,
{\em Exact Antiferromagnetic Ground States of Quantum Spin Chains}\/,
Europhys. Lett. {\bf 10}, 633--637 (1989).

\bibitem{FNW1992}
M. Fannes, B. Nachtergaele, and R.F. Werner,
{\em Finitely correlated states on quantum spin chains}\/,
Commun. Math. Phys. {\bf 144}, 443--490 (1992).
\\\url{https://projecteuclid.org/euclid.cmp/1104249404}

\bibitem{denNijsRommelse}
M. den Nijs and K. Rommelse,
{\em Preroughening transitions in crystal surfaces and valence-bond phases
in quantum spin chains}\/,
Phys. Rev. B {\bf 40}, 4709--4734 (1989).

\bibitem{KennedyTasaki1992}
T. Kennedy and H. Tasaki,
{\em Hidden symmetry breaking and the Haldane phase in $S=1$ quantum spin
chains}\/,
Commun. Math. Phys. {\bf 147}, 431--484 (1992).
\\\url{https://projecteuclid.org/euclid.cmp/1104250747}

\bibitem{Oshikawa1992}
M. Oshikawa,
{\em Hidden $\bbZ_2\times\bbZ_2$ symmetry in quantum spin chains with
arbitrary integer spin}\/,
J. Phys.: Condens. Matter {\bf 4}, 7469--7488 (1992).

\bibitem{GuWen2009}
Z.-C. Gu and X.-G. Wen,
{\em Tensor-entanglement-filtering renormalization approach and
symmetry-protected topological order}\/,
Phys. Rev. B {\bf 80}, 155131 (2009).
\\\url{https://arxiv.org/abs/0903.1069}

\bibitem{Pollmann2010}
F. Pollmann, A.M. Turner, E. Berg, and M. Oshikawa,
{\em Entanglement spectrum of a topological phase in one dimension}\/,
Phys. Rev. B {\bf 81}, 064439 (2010).
\\\url{https://arxiv.org/abs/0910.1811}

\bibitem{Pollmann2012}
F. Pollmann, E. Berg, A.M. Turner, and M. Oshikawa,
{\em Symmetry protection of topological phases in one-dimensional quantum
spin systems}\/,
Phys. Rev. B {\bf 85}, 075125 (2012).
\\\url{https://arxiv.org/abs/0909.4059}

\bibitem{PG2008}
D. P\'erez-Garc\'ia, M.M. Wolf, M. Sanz, F. Verstraete, and J.I. Cirac,
{\em String order and symmetries in quantum spin lattices}\/,
Phys. Rev. Lett. {\bf 100}, 167202 (2008).
\\\url{https://arxiv.org/abs/0802.0447}

\bibitem{ChenGuWen2011}
X. Chen, Z.-C. Gu, and X.-G. Wen,
{\em Classification of gapped symmetric phases in one-dimensional spin
systems}\/,
Phys. Rev. B {\bf 83}, 035107 (2011).
\\\url{https://arxiv.org/abs/1008.3745}

\bibitem{Schuch2011}
N. Schuch, D. P\'erez-Garc\'ia, and I. Cirac,
{\em Classifying quantum phases using matrix product states and projected
entangled pair states}\/,
Phys. Rev. B {\bf 84}, 165139 (2011).
\\\url{https://arxiv.org/abs/1010.3732}

\bibitem{BachmannNachtergaele2014}
S. Bachmann and B. Nachtergaele,
{\em On gapped phases with a continuous symmetry and boundary operators}\/,
J. Stat. Phys. {\bf 154}, 91--112 (2014).
\\\url{https://arxiv.org/abs/1307.0716}

\bibitem{Tasaki2018}
H. Tasaki,
{\em Topological phase transition and $\bbZ_2$ index for $S=1$ quantum
spin chains}\/,
Phys. Rev. Lett. {\bf 121}, 140604 (2018).
\\\url{https://arxiv.org/abs/1804.04337}

\bibitem{Ogata2020}
Y. Ogata,
{\em A $\bbZ_2$-index of symmetry protected topological phases with time
reversal symmetry for quantum spin chains}\/,
Commun. Math. Phys. {\bf 374}, 705--734 (2020).
\\\url{https://arxiv.org/abs/1810.01045}

\bibitem{Matsui2001}
T. Matsui,
{\em The split property and the symmetry breaking of the quantum spin
chain}\/,
Commun. Math. Phys. {\bf 218}, 393--416 (2001).

\bibitem{Matsui2013}
T. Matsui,
{\em Boundedness of entanglement entropy and split property of quantum
spin chains}\/,
Rev. Math. Phys. {\bf 25}, 1350017 (2013).
\\\url{https://arxiv.org/abs/1109.5778}

\bibitem{Hastings2007}
M.B. Hastings,
{\em An area law for one-dimensional quantum systems}\/,
J. Stat. Mech. P08024 (2007).
\\\url{https://arxiv.org/abs/0705.2024}

\bibitem{Ogata2021}
Y. Ogata,
{\em A classification of pure states on quantum spin chains satisfying
the split property with on-site finite group symmetries}\/,
Trans. Amer. Math. Soc. Ser. B {\bf 8}, 39--65 (2021).
\\\url{https://doi.org/10.1090/btran/51}

\bibitem{OgataSPT}
Y. Ogata,
{\em Classification of symmetry protected topological phases in quantum
spin chains}\/,
in {\em Current Developments in Mathematics 2020},
pp.~41--104 (International Press, 2022).
\\\url{https://doi.org/10.4310/CDM.2020.v2020.n1.a2}

\bibitem{OgataICM}
Y. Ogata,
{\em Classification of gapped ground state phases in quantum spin systems}\/,
in {\em International Congress of Mathematicians 2022},
Vol.~5, pp.~4142--4161 (EMS Press, 2023).
\\\url{https://ems.press/books/standalone/277/5549}

\bibitem{BMNS2012}
S. Bachmann, S. Michalakis, B. Nachtergaele, and R. Sims,
{\em Automorphic equivalence within gapped phases of quantum lattice
systems}\/,
Commun. Math. Phys. {\bf 309}, 835--871 (2012).
\\\url{https://arxiv.org/abs/1102.0842}

\bibitem{MoonOgata2020}
A. Moon and Y. Ogata,
{\em Automorphic equivalence within gapped phases in the bulk}\/,
J. Funct. Anal. {\bf 278}, 108422 (2020).
\\\url{https://arxiv.org/abs/1906.05479}

\bibitem{LSM1961}
E. Lieb, T. Schultz, and D. Mattis,
{\em Two soluble models of an antiferromagnetic chain}\/,
Ann. Phys. {\bf 16}, 407--466 (1961).

\bibitem{AffleckLieb1986}
I. Affleck and E.H. Lieb,
{\em A proof of part of Haldane's conjecture on spin chains}\/,
Lett. Math. Phys. {\bf 12}, 57--69 (1986).

\bibitem{NakamuraTodo2002}
M. Nakamura and S. Todo,
{\em Order Parameter to Characterize Valence-Bond-Solid States in Quantum
Spin Chains}\/,
Phys. Rev. Lett. {\bf 89}, 077204 (2002).
\\\url{https://arxiv.org/abs/cond-mat/0112377}

\bibitem{TasakiLocalLSM}
H. Tasaki,
{\em Lieb--Schultz--Mattis theorem with a local twist for general
one-dimensional quantum systems}\/,
J. Stat. Phys. {\bf 170}, 653--671 (2018).
\\\url{https://arxiv.org/abs/1708.05186}

\bibitem{Tasaki2023}
H. Tasaki,
{\em Rigorous index theory for one-dimensional interacting topological
insulators}\/,
J. Math. Phys. {\bf 64}, 041903 (2023).
\\\url{https://arxiv.org/abs/2111.07335}

\bibitem{Huang2022}
H.-Y. Huang, R. Kueng, G. Torlai, V.V. Albert, and J. Preskill,
{\em Provably efficient machine learning for quantum many-body problems}\/,
Science {\bf 377}, eabk3333 (2022).
\\\url{https://arxiv.org/abs/2106.12627}

\bibitem{OpenAIPeriodic}
OpenAI,
{\em The periodic spin-one Haldane gap}\/,
OpenAI Math Release preprint, September 24, 2026.
\\\href{https://github.com/openai/math/blob/fd4aeeb2ee4fc729c18d98444fed42fd0529eeeb/preprints/The-periodic-spin-one-Haldane-gap-September-24-2026/paper.pdf}
{\texttt{OpenAI Math Release: periodic spin-one Haldane gap}}

\bibitem{OpenAIBoundary}
OpenAI,
{\em A boundary-field gap for the spin-one Heisenberg chain}\/,
OpenAI Math Release preprint, September 24, 2026.
\\\href{https://github.com/openai/math/blob/fd4aeeb2ee4fc729c18d98444fed42fd0529eeeb/preprints/A-boundary-field-gap-for-the-spin-one-Heisenberg-chain-September-24-2026/paper.pdf}
{\texttt{OpenAI Math Release: boundary-field gap}}

\bibitem{Tasaki2025}
H. Tasaki,
{\em Ground state of the $S=1$ antiferromagnetic Heisenberg chain is
topologically nontrivial if gapped}\/,
Phys. Rev. Lett. {\bf 134}, 076602 (2025).
\\\url{https://arxiv.org/abs/2407.17041}

\bibitem{BR1}
O. Bratteli and D.W. Robinson,
{\em Operator Algebras and Quantum Statistical Mechanics 1:
$C^*$- and $W^*$-Algebras. Symmetry Groups. Decomposition of States}\/,
2nd ed. (Springer, 1987).


\bibitem{BR2}
O. Bratteli and D.W. Robinson,
{\em Operator Algebras and Quantum Statistical Mechanics 2:
Equilibrium States. Models in Quantum Statistical Mechanics}\/,
2nd ed. (Springer, 1997).


\bibitem{Naaijkens2017}
P. Naaijkens,
{\em Quantum Spin Systems on Infinite Lattices: A Concise Introduction}\/,
Lecture Notes in Physics, Vol.~933 (Springer, 2017).
\\\url{https://arxiv.org/abs/1311.2717}


\bibitem{TasakiLSMReview}
H. Tasaki,
{\em The Lieb--Schultz--Mattis theorem: A topological point of view}\/,
in R.L. Frank, A. Laptev, M. Lewin, and R. Seiringer, eds.,
{\em The Physics and Mathematics of Elliott Lieb}, Vol.~2,
pp.~405--446 (EMS Press, 2022).
\\\url{https://arxiv.org/abs/2202.06243}

\bibitem{Matsui1997}
T. Matsui,
{\em Translational Symmetry Breaking and Soliton Sectors for Massive
Quantum Spin Models in $1+1$ Dimensions}\/,
Commun. Math. Phys. {\bf 189}, 127--144 (1997).

\bibitem{TechnicalNote}
H. Tasaki,
{\em Twist indices for quantum spin chains: A technical note on cut
 dependence and higher-rank symmetries}\/,
technical note in preparation.

\bibitem{SuYaoFurusaki}
H. Su, Y. Yao, and A. Furusaki,
{\em Exact quantization of topological order parameter in $\mathrm{SU}(N)$
spin models, $N$-ality transformation and ingappabilities}\/,
Phys. Rev. Lett. {\bf 133}, 266705 (2024).
\\\url{https://arxiv.org/abs/2406.05407}

\end{thebibliography}
\end{document}